\documentclass[]{article} 
\usepackage{mathtools}
\usepackage{float}
\usepackage[utf8]{inputenc}
\usepackage{verbatim}
\usepackage{smartdiagram}
\usepackage{todonotes}
\usepackage{caption}
\usepackage{amsthm}

\usepackage{graphicx}
\usepackage{multicol}
\usepackage{footmisc}
\usepackage{cite}
\usepackage{url}

\usepackage{enumerate}

\usepackage{amsmath, amssymb, amsfonts}

\usepackage[linesnumbered, boxed, lined, ruled,vlined]{algorithm2e}

\usetikzlibrary{arrows.meta,positioning}
\usetikzlibrary{shapes.geometric}
\usetikzlibrary{shapes,fit}

\newtheorem{definition}{Definition}[section]
\newtheorem{theorem}[definition]{Theorem}
\newtheorem{lemma}[definition]{Lemma}
\newtheorem{corollary}[definition]{Corollary}

\newcommand{\N}{\mathbb{N}}
\newcommand{\Z}{\mathbb{Z}}

\newcommand{\rewardsum}{\sum_{i: A_i\subseteq S} a_i}
\newcommand{\penaltysum}{\sum_{j: B_j\cap S \neq \emptyset} b_j}

	\title{On Reward-Penalty-Selection Games}
\author{N. Gräf \and T. Heller\footnote{\texttt{till.heller@itwm.fraunhofer.de}, corresponding author} \and S.O. Krumke}

\begin{document}
           	
	\maketitle              
	
	\begin{abstract}
The Reward-Penalty-Selection Problem (RPSP) can be seen as a combination of the Set Cover Problem (SCP) and the Hitting Set Problem (HSP). Given a set of elements, a set of reward sets, and a set of penalty sets, one tries to find a subset of elements such that as many reward sets as possible are covered, i.e. all elements are contained in the subset, and at the same time as few penalty sets as possible are hit, i.e. the intersection of the subset with the penalty set is non-empty. In this paper we define a cooperative game based on the RPSP where the elements of the RPSP are the players. We prove structural results and show that RPS games are convex, superadditive and totally balanced. Furthermore, the Shapley value can be computed in polynomial time. In addition to that, we provide a characterization of the core elements as a feasible flow in a network graph depending on the instance of the underlying RPSP. By using this characterization, a core element can be computed efficiently.
	\end{abstract}

\section{Introduction}
For a company that specializes in the further processing of raw materials into various end products, two strategic questions are of particular interest. Firstly, which products should be produced and secondly, which materials should be purchased for them. In this case, different materials can be purchased together, whereby the full price is incurred as soon as only one of the materials from a package is required. On the other hand, end products can only be produced if all the materials required for them are available - only in this case can a profit be made. 

This problem can be formalized as the \emph{reward-penalty-selection problem} (RPSP) which was recently introduced in \cite{heller2021reward}. It can be viewed as a combination of two classical, well-known combinatorial problems, the set cover problem and the hitting set problem (cf. \cite{garey1979computers}). The ground set of elements~$N$ is given by the set of materials. The set of reward sets~$A\coloneqq \{A_i | A_i\subseteq N\}$ with corresponding rewards~$a_i\in\N$ and a set of penalty sets~$B\coloneqq \{B_i | B_i\subseteq N\}$ with corresponding penalties~$b_i\in\N$ are given by the end products and raw material packages with their respective price.  We say that a set is \emph{covered} if all elements of said set are chosen and that a set is \emph{hit} if at least one element is chosen. Now we aim to find a subset of elements~$S$ such that the profit function
\begin{align}
\rewardsum - \penaltysum \label{eq: rpsp objective function}
\end{align}
is maximized. In the example given above, this corresponds to selecting a subset of the raw materials such that the net profit defined by the profit obtained by selling the final products minus the cost for the necessary raw material packages is maximized.  For an analysis on the complexity of the RPSP and its variants, as well as algorithmic approaches, we refer to \cite{heller2021reward,heller2021phd}.

The question how ``fair prices'' for the materials could look like can be tackled by considering the problem from the viewpoint of game theory, which is the central point of this paper.
% The task to distribute the achieved profit fairly to all involved players can be seen as profit distribution. 
For this, we view each material as a player in a cooperative game with characteristic function as in~\eqref{eq: rpsp objective function}.  This leads us to \emph{reward-penalty-selection games} (RPS games) which we define formally later. 

Cooperative games are often used to find ``fair'' solutions (cf.~\cite{nisanalgorithmic}) for settings where multiple players cooperate to obtain a profit together. In general, a cooperative game consists of a group of players and a \emph{characteristic function}, which maps every subset of players to its obtainable profit. A \emph{solution} to such a cooperative game is given by a \emph{payment vector} that stores the profit for each player. Properties, that are widely considered as fair, are the following ones. We say that a payment vector fulfills \emph{efficiency} (EFF), if the obtained profit of the group of all players is fully distributed. If all entries in the vector are greater or equal to the profit the corresponding player can obtain by itself, \emph{individual rationality} (IR) is satisfied. \emph{Coalitional rationality} (CR) is fulfilled if the same holds for all subsets of the players, i.e. the sum of the payments to players is greater or equal to the obtained profit by this subset of players. The \emph{core} consists of all those payment vectors that fulfill (EFF), (IR) and (CR) (cf. \cite{nisanalgorithmic}). For the rest of the paper, we focus on core vectors, such as the \emph{Shapley value} (cf. \cite{shapley1953value}), but we note that there exists other payment vectors that are considered as fair, for instance the \emph{egalitarian allocation} (cf. \cite{dutta1989concept, koster1999weighted}). One of the main results of this paper is a characterization of the core vectors via flows in an associated network.

The remainder of the paper is structured as follows. In Section~\ref{sec: rps game} we give a definition of an RPS game. We also provide basic results regarding the properties of such games. In Section~\ref{sec: core characterization} we prove a characterization of the core elements of an RPS game as a network flow in a suitable network graph. We then conclude with a short outlook.

\section{The Reward-Penalty-Selection Game}\label{sec: rps game}

We start with a formal definition of \emph{reward-penalty-selection (RPS) games}:

\begin{definition}[Reward-Penalty-Selection Games]	
	Let $N\coloneqq\{1,\dots,n\}$ be the set of players. Further, let $\mathcal{A}\coloneqq\{A_1,\dots,A_k\}\subseteq 2^N$ be the set of non-empty reward sets with rewards $a_1,\dots,a_k\in\mathbb{N}$ and $\mathcal{B}\coloneqq\{B_1,\dots,B_l\}\subseteq 2^N$ be the set of non-empty penalty sets with penalties~$b_1,\dots,b_l\in\mathbb{N}$. The game~$(N,v)$ with characteristic function~$v\colon 2^N \rightarrow \Z$ defined by
	\begin{align*}
	v(S) \coloneqq \rewardsum - \penaltysum
	\end{align*}
	is called \emph{reward-penalty-selection game (RPS game)}. 
\end{definition}

\begin{theorem}[Convexity of RPS games]\label{thm: convexity RPS games}
	Every RPS game~$(N,v)$ is convex, i.e. the characteristic function satisfies $v(S) + v(T) \leq v(S\cup T) + v(S\cap T)$ for arbitrary subsets~$S,T\subseteq N$ of players. 
\end{theorem}
\begin{proof}
	Note that it is sufficient to check the cases of a single reward or a single penalty set since the sum of convex games is again convex. Let~$(N,v)$ be an RPS game consisting of one non-empty reward set $A\subseteq N$ with reward~$a$. Further, let $S,T\subseteq N$ be arbitrary sets of players.
	
	If $v(S\cup T) = 0$, then $A$ is not contained in $S\cup T$ and therefore not either in~$S$ and~$T$. This implies that
	\begin{align*} 
	v(S\cup T) = 0 = v(S) + v(T) - v(S\cap T),
	\end{align*}
	thus, convexity is fulfilled. 
	
	If $v(S\cup T) = a$, we have to distinguish between two cases. Either $A\subseteq S\cap T$, which implies 
	\begin{align*}
	v(S\cup T) =a = a + (a-a) = v(S) + v(T) - v(S\cap T),	
	\end{align*}
	or $A \not\subseteq S\cap T$. In this case, $A$ cannot be contained in both $S$ and $T$ and therefore
	\begin{align*}
	v(S\cup T) = a = a-0 \geq v(S) + v(T) - v(S\cap T). 	
	\end{align*}
	The case of an RPS game with a single penalty set can be shown in the same way.
\end{proof}

Furthermore, by a similar argument as in the proof of convexity, we obtain the following corollary.

\begin{corollary}
	RPS games are superadditive, i.e. the characteristic function satisfies $v(T\cup S) \geq v(S) + v(T)$ for all $S,T\subseteq N$ with $S\cap T = \emptyset$.\qed
\end{corollary}

The main purpose of such a cooperative game is to find a fair profit distribution of the total profit among all players. For this, we define the \emph{payment vector} as the vector~$p\in\mathbb{R}^N$ whose $i$th entry~$p_i$ is the payment for player~$i$. We define the payment to a coalition~$S\subseteq N$ as the sum over the payments to the single players, i.e. $p(S) \coloneqq \sum_{i\in S} p_i$. Given such a payment vector, we now formally define desired properties as introduced above. A payment vector~$p$ fulfills \emph{efficiency} (EFF), if $\sum_{i\in N} p_i = v(N)$ holds. Furthermore, if $p_i\geq v(\{i\})$ for all~$i\in N$, we say $p$ fulfills \emph{individual rationality} (IR). The extension of (IR) to coalitions is called \emph{coalitional rationality} (CR). We say a payment vector fulfills (CR) if each coalition is guaranteed at least the value the players could obtain by themselves, i.e. $\sum_{i\in S} p_i \geq v(S)$.

Now, the \emph{core} is a well known solution to such cooperative games and defined as the set of payments, which fulfills (EFF), (IR) and (CR). Using the convexity of RPS games, it follows that an element in the core can be computed in polynomial time (cf. \cite{nisanalgorithmic}). Thus, it follows immediately:

\begin{corollary}\label{cor:poly}
  RPS games are balanced (cf. \cite{shapley1965balanced}), i.e. the core of an RPS game is never empty. A core element can be computed in polynomial time.
  \qed \end{corollary}

The following two lemmas present structural insights into RPS games.

\begin{lemma}
	RPS games are totally balanced.
\end{lemma}
\begin{proof}
	Each subgame of an RPS game is itself an RPS game. Thus, since each RPS game is balanced, so are its subgames and therefore RPS games are totally balanced.
\end{proof}

\begin{lemma}
  Suppose we are given an RPS game where all reward and penalty sets consist of exactly one player. Then the core consists of a singleton.  \end{lemma} \begin{proof}
	Suppose $|A_i| = |B_j| = 1$ for all~$i$, $j$. Further, let $p$ be a core vector. Hence, $p$ has to satisfy $p_k \geq v(\{k\}) = \sum_{i:\{k\} = A_i} a_i - \sum_{j: \{k\} = B_j} b_j$. Since a core vector also has to fulfill $p(N) = v(N)$, we get
	\begin{align*}
	\sum_{k\in N} \left( \sum_{i:\{k\} = A_i} a_i - \sum_{j: \{k\} = B_j} b_j \right)  &\leq \sum_{k\in N} p_k \\
	& = p(N)\\
	& = v(N)\\
	& = \sum_{k\in N} \left( \sum_{i:\{k\} = A_i} a_i - \sum_{j: \{k\} = B_j} b_j \right). 
	\end{align*}
	Thus, $p(N) = \sum_{k\in N} \left( \sum_{i:\{k\} = A_i} a_i - \sum_{j: \{k\} = B_j} b_j \right)$ is the only core vector.
\end{proof}

A famous core payment is the \emph{Shapley value} (cf. \cite{shapley1953value}). Formally, for a cooperative game~$(N,v)$ the Shapley value of a player~$k$ is defined as
\begin{align*}
	\phi_v(k) \coloneqq \sum_{S\subseteq N\backslash\{k\}} \frac{|S|!(n - |S| - 1)!}{n!}(v(S\cup\{k\}) - v(S)).
\end{align*}
This can be interpreted as the average marginal profit the player~$k$ adds to an existing coalition~$S$ when joining the coalition over all possible permutations. The next theorem states that the Shapley value can be computed efficiently for RPS games. 

\begin{theorem}\label{thm: characterization of RPS games}
	Given an RPS game, the Shapley value~$\phi_v$ for a player~$k$ is given by
	\begin{align}
	\phi_v(k) \coloneqq \sum_{i\in N} \frac{a_i}{|A_i|} - \sum_{j\in N} \frac{b_j}{|B_j|} \label{eq: shapley value}
	\end{align}
	for $k=1,\dots, n$ and, thus, can be computed efficiently.
\end{theorem}
\begin{proof}
For each reward set~$A_i\in\mathcal{A}$ a player~$k$ adds value $a_i$ to the coalition value if and only if all other players of~$A_i$ are already contained in the coalition, i.e. in $\frac{1}{|A_i|}$ of all cases. Hence, on average each player in~$A_i$ contributes $\frac{a_i}{|A_i|}$ to the value of the coalition. On the other side, for each penalty set~$B_j\in\mathcal{B}$, a player~$k$ incurs cost~$b_j$ if and only if she enters the coalition first. Again, on average, each player in $B_j$ incurs cost~$\frac{b_j}{|B_j|}$ to the coalition. By summing up over all reward and penalty sets we obtain~\eqref{eq: shapley value}.
		
In order to compute the Shapley value of a player, the contribution of each player is initialized with zero. Now we iterate over all reward and penalty sets and update the contribution of each player according to~\eqref{eq: shapley value}.
\end{proof}

% We note that the computation can be done in time linear in the size of the bipartite graph with the shores $A$; $B$; $N$. and edges.. 

We now investigate the relation of RPS games to the set of convex games.

\begin{lemma}\label{lem: modeling with RPS games}
	The following statements are true:
	\begin{enumerate}[(i)]
        \item Every convex cooperative game with three players can be modeled as an RPS game.
        \item The set of all convex  cooperative games is a strict superset of the set of RPS games, meaning that there exist convex cooperative games with four players that cannot be modeled as an RPS game.
	\end{enumerate}
\end{lemma}

\begin{proof}
	\begin{enumerate}[(i)]
		\item Let $(N,v)$ be a convex cooperative game with player set~$N=\{1,2,3\}$. W.l.o.g. we can assume that the value of a singleton is always zero since we can add singleton reward and penalty sets while keeping convexity. Let $d\coloneqq\min\{v(\{1,2\}), v(\{1,3\}), v(\{2,3\})\}$ be the minimal value of a coalition consisting of two players. By convexity, $d$ is non-negative. If $d$ is greater than $0$, we add a penalty set $B = \{1,2,3\}$ with penalty $b = d$. In order to obtain the amount given by the characteristic function~$v$, we add reward sets
		\begin{align*}
		A_1 = \{1\} &\text{ with } a_1 = d \\
		A_2 = \{2\} &\text{ with } a_2 = d \\
		A_3 = \{3\} &\text{ with } a_3 = d \\
		A_4 = \{1,2\} &\text{ with } a_4 = v(\{1,2\}) - d \\
		A_5 = \{2,3\} &\text{ with } a_5 = v(\{2,3\}) - d \\
		A_6 = \{1,3\} &\text{ with } a_6 = v(\{1,3\}) - d
		\end{align*}
		Note that all rewards are non-negative by the definition of~$d$. Finally, we add $A_7 = \{1,2,3\} \text{ with } a_7 = v(\{1,2,3\}) - \sum_{i=1}^{6} a_i - b$. With this, an arbitrary game consisting of three players can be modeled as an RPS game. 
		
		\item Let $(N,v)$ be the convex cooperative  game with player set~$N=\{1,2,3,4\}$ and characteristic function
		\begin{align*} 
		v(S) = \begin{cases}
		0, \qquad \quad \text{ if } |S|\leq 2, \\
		1, \qquad \quad \text{ if } |S| = 3, \\
		2, \qquad \quad \text{ if } S = N. \\
		\end{cases} 	
		\end{align*}
		Suppose $(N,v)$ can be modeled by an RPS game. Let $i,j\in N$ be two different players. It holds that $v(\{i,j\}) = 0 = v(i) + v(j)$. Since $v(\{i,j\}) \geq v(i) + v(j)$ is always true due to convexity, there cannot be a penalty set containing both $i$ and $j$, because otherwise for the coalition the penalty incurs once while the penalty incurs for both singletons~$\{i\}, \{j\}$.  
		
		Adding more sets, the gap only increases which does not help in order to obtain $v(\{i,j\}) = v(i) + v(j)$. Hence, no penalty set of size two or larger can be used to model the game. As every two player coalition receives value zero, reward sets of size exactly two cannot exist in said game at the same time. In order to generate a profit of one for all three-player coalitions, we have to add a reward set of value one for each of those coalitions. But then, since there are four such coalitions, the RPS game must grant a profit of at least four to the grand coalition --- a contradiction. 
	\end{enumerate}
\end{proof}
\section{Characterization of Core Elements}\label{sec: core characterization}
In this section we give a characterization of core elements of an instance of an RPS game. In order to do this, we define a profit sharing graph and prove that any feasible flow in this graph of a certain flow value induces a core vector and vice versa.  As a byproduct of this characterization we obtain an alternative proof to the polynomial time computability of a core element stated in Corollary~\ref{cor:poly}.  We assume that the reader is familiar with the basics of network flows~\cite{ahuja1988network}.

Our approach of a characterization of core elements as feasible flows is based on results of Ackermann et al.~(cf. \cite{ackermann2014modeling}). First, we define the \emph{profit sharing graph} for RPS games.  
\begin{definition}[Profit Sharing Graph for RPS Games]
	Let $(N,v)$ be an RPS game with player set~$N=\{1,2,\dots,n\}$, a collection of reward sets~$\mathcal{A} = \{A_1,\dots,A_k\}$ and a collection of penalty sets~$\mathcal{B} = \{B_1,\dots,B_l\}$. The \emph{profit sharing graph} for $(N,v)$ is given by the directed graph~$G=(V,E)$ with nodes
	\begin{align*} 
	V\coloneqq \{s,t,\overline{s},\overline{t}\}\cup N \cup \mathcal{A} \cup \mathcal{B},	
	\end{align*}
	and edges
	\begin{alignat*}{2}
	E \coloneqq& \{(s,A): A\in\mathcal{A}\} \cup \{(B,t): B\in\mathcal{B}\} \cup \{(s,\overline{s}), (\overline{t}, t)\} \cup \\
	&	\{(\overline{s},n): n\in N\} \cup \{(n,\overline{t}): n\in N\} \, \cup \\
	&	\{(A,i): i\in A, A\in\mathcal{A}\} \cup	\{(i,B): i\in B, B\in\mathcal{B}\} \cup \{(\overline{s}, \overline{t})\}.
	\end{alignat*}
	We set the edge capacities to be given by the function~$c\colon E\to \Z$ defined by
	\begin{align*}
	c(e) \coloneqq \begin{cases}
	a_i,& \text{for $e=(s,A_i)$ and $A_i\in\mathcal{A}$}\\
	b_j, & \text{for $e=(B_j,t)$ and $B_j\in\mathcal{B}$}\\
	\sum_{j=1}^l b_j,& \text{for $e=(s,\overline{s})$} \\
	\sum_{i=1}^k a_i, & \text{for $e=(\overline{t},t)$} \\
	\infty,& \text{otherwise.}
	\end{cases}
	\end{align*}
\end{definition}

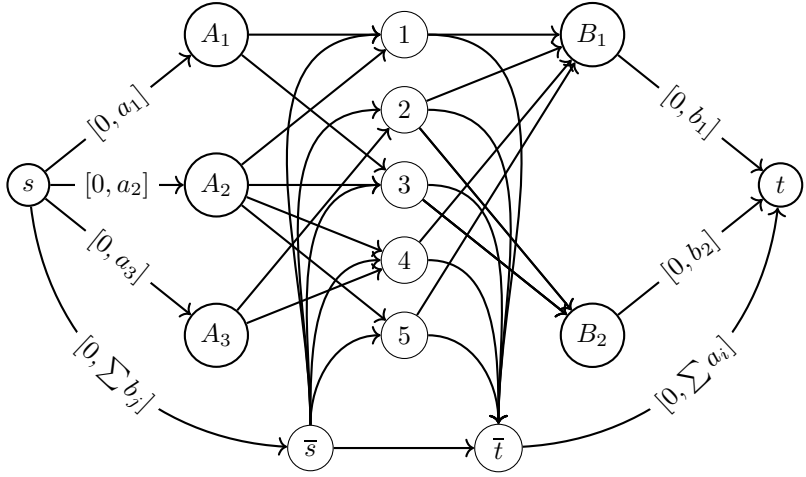
\begin{figure}
	\centering
	\begin{tikzpicture}[every node/.style={fill=white,rectangle}, every edge/.style={draw=black,very thick}]
	\begin{scope}[every node/.style={circle,thick,draw}]
	\node (s) at (0,3) {$s$};
	\node (A1) at (2.5,5) {$A_1$};
	\node (A2) at (2.5,3) {$A_2$};
	\node (A3) at (2.5,1) {$A_3$};
	
	\node (B1) at (7.5,5) {$B_1$};
	\node (B2) at (7.5,1) {$B_2$};
	
	\node (t) at (10, 3) {$t$};
	\end{scope}
	
	\begin{scope}
	[every node/.style={circle, draw}]
	\node (n1) at (5,5) {$1$};
	\node (n2) at (5,4) {$2$};
	\node (n3) at (5,3) {$3$};
	\node (n4) at (5,2) {$4$};
	\node (n5) at (5,1) {$5$};
	
	\node (s1) at (3.75, -0.5) {$\overline{s}$};
	\node (t1) at (6.25, -0.5) {$\overline{t}$};
	\end{scope}
	
	\begin{scope}[
	every node/.style={fill=white,rectangle,sloped},
	every edge/.style={draw=black,thick}]
	
	\path [->] (s) edge node {$[0,a_1]$} (A1);
	\path [->] (s) edge node{$[0,a_2]$} (A2);
	\path [->] (s) edge node{$[0,a_3]$} (A3);
	
	\path [->] (A1) edge (n1);
	\path [->] (A2) edge (n1);
	\path [->] (A2) edge (n3);
	\path [->] (A3) edge (n2);
	\path [->] (A1) edge (n3);
	\path [->] (A2) edge (n5);
	\path [->] (A2) edge (n4);
	\path [->] (A3) edge (n4);
	
	\path [->] (n2) edge (B1);
	\path [->] (n2) edge (B2);
	\path [->] (n3) edge (B2);
	\path [->] (n1) edge (B1);
	\path [->] (n2) edge (B2);
	\path [->] (n3) edge (B2);
	\path [->] (n4) edge (B1);
	\path [->] (n5) edge (B1);
	\path [->] (n3) edge (B2);
	
	\path [->] (B1) edge node{$[0,b_1]$} (t);
	\path [->] (B2) edge node{$[0,b_2]$} (t); 
	
	\path [->] (t1) edge[bend left=-40] node{$[0,\sum a_i]$} (t); 
	\path [->] (s) edge[bend left=-40] node{$[0,\sum b_j]$} (s1); 
	\path [->] (s1) edge (t1);

	\path [->] (n1) edge[out=0, in=90] (t1); 
	\path [->] (n2) edge[out=0, in=90] (t1); 
	\path [->] (n3) edge[out=0, in=90] (t1); 	
	\path [->] (n4) edge[out=0, in=90] (t1); 	
	\path [->] (n5) edge[out=0, in=90] (t1); 
	
	\path [<-] (n1) edge[out=180, in=90] (s1); 
	\path [<-] (n2) edge[out=180, in=90] (s1); 
	\path [<-] (n3) edge[out=180, in=90] (s1); 	
	\path [<-] (n4) edge[out=180, in=90] (s1); 	
	\path [<-] (n5) edge[out=180, in=90] (s1); 
	
	\end{scope}
	\end{tikzpicture}
	\caption{Example of a \emph{profit sharing graph} for $N=\{1,2,3,4,5\}$.}\label{fig: profit sharing}
\end{figure}

Furthermore, we set the lower capacity bound~$l(e)$ of each edge to~0. An example of the profit sharing graph can be found in Figure~\ref{fig: profit sharing}. It is clear by construction that any feasible $s-t-$flow in the profit sharing graph of an RPS game with value~$H\coloneqq\sum_{i:A_i\in\mathcal{A}} a_i + \sum_{j:B_j\in\mathcal{B}} b_j$ fully exhausts all finite capacities. 

% For a feasible RPS game~$(N,v)$ and the corresponding profit sharing graph~$G$, we can restrict ourselves to the \emph{feasible profit sharing graph} where we delete all edges of the form~$(\overline{s}, k)$ for~$k\in N$.

%In order to show that a feasible flow exists, we make use of the Hofmann's Circulation Theorem (cf. \cite{ahuja1988network}) which states that for a given network graph~$G_H=(V_H\cup\{s,t\},E_H)$ there exists a feasible $s-t-$flow~$f_H$ if and only if
%\begin{align} \label{eq:1}
%b(S) + l(S) \leq c(S) \quad \text{ for all subsets~$S\subseteq V$}
%\end{align}
%is fulfilled where $b(S)$ denotes the sum of the balances of nodes in $S$, $l(S)$ denotes the sum of lower capacity bounds of edges from $V\backslash S$ to $S$ and $c(S)$ denotes the sum of capacities of edges from $S$ to $V\backslash S$. Clearly, if an edge with capacity~$\infty$ lies in the cut, condition~(\ref{eq:1}) is always fulfilled. For the profit sharing graph this means that if a player node~$k$ lies in~$S$, all connected reward and penalty nodes also lie in~$S$. With this, the Hofmann condition~(\ref{eq:1}) becomes
%\begin{align*}
%& b(S) + l(S) \leq c(S) \\
%\Leftrightarrow \qquad \qquad &\sum_{i: S\cap A_i \neq \emptyset} a_i - \sum_{B_j\in\mathcal{B}} b_j \leq \sum_{A_i\in\mathcal{A}} a_i - \sum_{j:S\cap B_j \neq \emptyset} b_j \\
%\Leftrightarrow \qquad \qquad& 0 \leq \sum_{i: S\cap A_i = \emptyset} a_i - \sum_{B_j: S\cap B_j = \emptyset}b_j, 
%\end{align*}
%which is exactly the characterization of a feasible RPS game. 

Now, let $(N,v)$ be an RPS game and $G=(V,E)$ the corresponding profit sharing graph. We define the payment~$p_i$ to a player~$i$ by 
\begin{align}
p_i \coloneqq f(i,\overline{t}) - f(\overline{s}, i). \label{eq: payment by flow}
\end{align}

The next theorem shows the connection between a feasible flow in the profit sharing graph and a core vector. 
\begin{theorem}[Core Elements]\label{thm: flow induces core}
	Any feasible flow with value~$H$ in the profit sharing graph of an RPS game defines a payment vector that fulfills the properties of efficiency (EFF), coalitional rationality (CR) and individual rationality (IR).
\end{theorem}
\begin{proof}
	We prove the claim by showing that a payment vector~$p$ defined by a feasible flow~$f$ with value~$H$ fulfills the above conditions, where the payment of a player~$i$ is is defined by \eqref{eq: payment by flow}. Recall that a feasible flow fully exhausts all finite capacities of the profit sharing graph.	
	\begin{enumerate}[(i)]
		\item Efficiency: For efficiency we need to show $p(N) = v(N)$. By definition we have 
		\begin{align*}
		p(N) 	&= \sum_{i\in N} p_i = \sum_{i\in N} f(i,\overline{t}) - f(\overline{s},i).
		\end{align*}
		By using flow conservation at each of the player nodes~$i$ this is equal to
		\begin{align*}
		\sum_{i\in N} \left(f(\mathcal{A}, i) - f(i,\mathcal{B})\right)	& = \sum_{i: A_i\in\mathcal{A}} a_i - \sum_{j: B_j\in\mathcal{B}} b_j = v(N).
		\end{align*}
		
		% We have seen above that there always exists a solution in the feasible profit sharing graph for feasible RPS games, i.e. $p_i\geq 0$ for all $i\in N$. Note that this is not true in general since the feasible profit sharing graph only has a solution for a feasible RPS game.
		
		\item Coalitional Rationality: Let $S\subseteq N$ be a subset of players. We want to show that the profit distributed to this subgroup~$S$ is greater or equal to the value of $S$, i.e. $p(S) \geq v(S)$. By definition we have
		\begin{align*}
		p(S) &= \sum_{i\in S} p_i = \sum_{i\in S} \left(f(i,\overline{t}) - f(\overline{s},i)\right)
		\end{align*}
		By using flow conservation this is equal to
		\begin{align*}
		\sum_{A\in\mathcal{A}}\sum_{i\in S} f(A,i) - \sum_{B\in\mathcal{B}}\sum_{i\in S} f(i,B)
		& = \sum_{i: A_i\cap S\neq \emptyset} f(A_i,S) - \sum_{j: B_j\cap S \neq \emptyset} f(S,B_j)\\
		& \geq \sum_{i: A_i\subseteq S} f(A_i,S) - \sum_{j: B_j\cap S \neq \emptyset} f(S,B_j)\\
		& = \rewardsum - \penaltysum \\
		& = v(S),
		\end{align*}
		since every feasible flow with value~$H$ fully exhausts all finite capacities. 
		\item Individual Rationality: Follows by the same argumentation as for (CR).
	\end{enumerate}
	Thus, all four properties are fulfilled.
\end{proof}

Conversely, the next theorem shows that given a core vector, one finds a corresponding feasible flow. 
\begin{theorem}[Core Elements II]\label{thm: core induces flow}
	Let $p$ be a core allocation of an RPS game~$(N,v)$. Then, there exists a feasible flow~$f$ with value~$H$ in the corresponding profit sharing graph that induces the allocation~$p$. 
\end{theorem}
\begin{proof}
	Let $G=(V,E)$ be the corresponding profit sharing graph of~$(N,v)$. We modify~$G$ by changing the infinite capacities to finite ones. Define $c_p$ as follows:
	\begin{align*}
	c_p(\overline{s},i) &= \begin{cases}
	p_i & \text{if $p_i \geq 0$} \\
	0 & \text{otherwise,}
	\end{cases}\\
	c_p(i, \overline{t}) &= \begin{cases}
	-p_i & \text{if $p_i < 0$} \\
	0 & \text{otherwise}
	\end{cases}\\
	\intertext{and}
	c_p(\overline{s},\overline{t}) &= \sum_{B_j\in\mathcal{B}}b_j - \sum_{i\in N} c(\overline{s}, i).
	\end{align*}
	We now define finite capacities $\operatorname{cap}\colon E\rightarrow\Z$ by
	\begin{align*}
	\operatorname{cap}(e) \coloneqq 
	\begin{cases}
	c(e), &\text{if $c(e)<\infty$}\\
	c_p(e), & \text{if $c(e)=\infty$.}
	\end{cases}
	\end{align*}
	With these finite capacities, the payment from or to a player~$i$ made by $p$ is representable in the flow network since $c_p(\overline{s},i) + c_p(i, \overline{t}) = |p_i|$ and $c_p(\overline{s}, i) - c_p(i,\overline{t}) = p_i$ always hold true.
	
	By the Max-Flow Min-Cut Theorem (cf. \cite{ahuja1988network}), the capacity of a minimum $s$-$t$-cut in $G$ is equal to the value of a feasible $s$-$t$-flow in~$G$. Therefore we show in the following that the capacity of a minimum cut is at least~$H$. Thus, let $X\subseteq V$ be an $s$-$t$-cut~in $G$, i.e. $s\in X$ and $t\notin X$. Since $X'=\{s\}$ is a cut with finite capacity, we know that all edges connecting~$X$ and $V\setminus X$ have finite capacity. That means for an arbitrary player~$i$ in $N\cap X$ all reward and penalty sets that contain~$i$ belong also to~$X$. Further, for any reward or penalty set that are contained in~$X$, all its set members are also contained in~$X$. With this, we get
	\begin{align*}
	\mathcal{A}\cap X = \{\,A: A\cap X \neq \emptyset, A\in\mathcal{A}\,\} = \{\,A: A\cap N \cap X \neq \emptyset, A\in\mathcal{A}\,\}
	\intertext{and}
	\mathcal{A}\backslash X = \{\,A: A\cap X = \emptyset, A\in\mathcal{A}\,\} = \{\,A: A\cap N \cap X = \emptyset, A\in\mathcal{A}\,\}.	
	\end{align*} 
	
	There are four possibilities depending on whether $\overline{s}\in X$ or $\overline{t}\in X$.
	\begin{enumerate}[(i)]
		\item Assume $\overline{s},\overline{t}\notin X$. Then, we get
		\begin{align*}
		\sum_{i\in N\cap X} \left( c_p(\overline{s}, i) + c_p(i,\overline{t}) \right) &= \sum_{i\in N\cap X} |p_i| \\
		&\geq \sum_{i\in N\cap X} p_i = p(N\cap X).
		\end{align*}
		By coalitional rationality, this is greater or equal to $v(N\cap X)$. Thus,
		\begin{align*}
		v(N\cap X) &= \sum_{i: A_i\subseteq N\cap X} a_i - \sum_{j: B_j\cap N\cap X \neq \emptyset} b_j \\
		& =  \sum_{i: A_i\cap N\cap X} a_i - \sum_{j: B_j\cap N\cap X \neq \emptyset} b_j \\
		& = \sum_{i: A_i\in \mathcal{A}\cap X} a_i - \sum_{j: B_j\in \mathcal{B}\cap X \neq \emptyset} b_j
		\end{align*}
		With the above calculation, we get as the capacity of the cut~$X$:
		\begin{align*}
		\operatorname{cap}(X) & = c(s,\overline{s}) + \sum_{i:A_i \in \mathcal{A}\backslash X} c(s,A_i) + \sum_{j:B_j\in \mathcal{B}\cap X} c(B_j,t) \\
		& \qquad \qquad \qquad + \sum_{i\in N\cap X} \left(c_p(\overline{s}, i) + c_p(i, \overline{t})\right) \\
		&\geq \sum_{1\leq j \leq |\mathcal{B}|} b_j + \sum_{i:A_i \in \mathcal{A}\backslash X} a_i + \sum_{j:B_j\in \mathcal{B}\cap X} b_j \\
		& \qquad \qquad \qquad + \left( \sum_{i: A_i\in \mathcal{A}\cap X} a_i - \sum_{j: B_j\in \mathcal{B}\cap X \neq \emptyset} b_j \right) \\
		& = \sum_{1\leq j \leq |\mathcal{B}|}b_j + \sum_{1\leq i \leq |\mathcal{A}|} a_i = H.
		\end{align*}
		\item Assume $\overline{s},\overline{t}\in X$.  Then
		\begin{align*}
		\sum_{i\in N\backslash X} \left( c_p(\overline{s}, i) + c_p(i,\overline{t}) \right) &= \sum_{i\in N\backslash X} |p_i| \\
		&\geq \sum_{i\in N\backslash X} -p_i \\
		& = -p(N\backslash X)\\
		& = -p(N) + p(N\cap X) \\
		& \geq -v(N) + v(N\cap X) \\
		& = \sum_{1\leq j \leq |\mathcal{B}|} b_j - \sum_{1\leq i \leq |\mathcal{A}|}a_i + \sum_{i: A_i\cap N\cap X} a_i\\
		& \qquad \qquad \qquad - \sum_{j: B_j\cap N\cap X \neq \emptyset} b_j \\
		& = \sum_{j: B_j\cap N\cap X = \emptyset} b_j - \sum_{i: A_i\cap (N\backslash X) \neq \emptyset} a_i \\
		& \sum_{j:B_j\in \mathcal{B}\backslash X} b_j - \sum_{i: A_i\in \mathcal{A}\backslash X} a_i.
		\end{align*}
		With this, we get
		\begin{align*}
		\operatorname{cap}(X) &= c(\overline{t},t) + \sum_{i:A_i \in \mathcal{A}\backslash X} c(s,A_i) + \sum_{j:B_j\in \mathcal{B}\cap X} c(B_j,t) \\
		& \qquad \qquad \qquad + \sum_{i\in N\backslash X} \left(c_p(\overline{s}, k) + c_p(i, \overline{t})\right) \\
		&= \sum_{1\leq i \leq |\mathcal{A}|}a_i + \sum_{i: A_i\in \mathcal{A}\backslash X}a_i + \sum_{j:B_j\in \mathcal{B}\cap X} + \sum_{i\in N\cap X} \left(c_p(\overline{s},i) + c_p(i,\overline{t})\right)\\
		&\geq \sum_{1\leq i \leq |\mathcal{A}|}a_i + \sum_{i: A_i\in \mathcal{A}\backslash X}a_i + \sum_{j:B_j\in \mathcal{B}\cap X} + \left(\sum_{j:B_j\in \mathcal{B}\cap X} b_j - \sum_{i: A_i\in \mathcal{A}\cap X} a_i\right)\\
		& = \sum_{1\leq i \leq |\mathcal{A}|}a_i + \sum_{1\leq j \leq |\mathcal{B}|} b_j = H.
		\end{align*}
		\item Assume $\overline{s}\in X$ and $\overline{t}\notin X$. Thus, the following holds by similar calculations as above.
		\begin{align*}
		\operatorname{cap}(X) &= c(\overline{s},\overline{t}) + \sum_{i: A_i\in \mathcal{A}\backslash X}c(s,A_i) + \sum_{j:B_j\in \mathcal{B}\cap X}c(B_j,t) \\
		& \qquad \qquad \qquad + \sum_{i\in N\cap X}c_p(i,\overline{t}) + \sum_{i\in N\backslash X}c_p(\overline{s},i)\\
		&= \sum_{1\leq j \leq |\mathcal{B}|} b_j - \sum_{i\in N} c_p(\overline{s},i) + \sum_{i: A_i\in \mathcal{A}\backslash X}a_i + \sum_{j:B_j\in \mathcal{B}\cap X}b_j \\
		& \qquad \qquad \qquad + \sum_{i\in N\cap X}c_p(i,\overline{t}) + \sum_{i\in N\backslash X}c_p(\overline{s},i)\\
		&= \sum_{1\leq j \leq |\mathcal{B}|}b_j + \sum_{i: A_i\in \mathcal{A}\backslash X}a_i + \sum_{j:B_j\in \mathcal{B}\cap X}b_j + \sum_{i\in N\cap X}c_p(i,\overline{t}) \\
		& \qquad \qquad \qquad - \sum_{i\in N\cap X} c_p(\overline{s},i) \\
		&= \sum_{1\leq j \leq |\mathcal{B}|}b_j + \sum_{i: A_i\in \mathcal{A}\backslash X}a_i + \sum_{j:B_j\in \mathcal{B}\cap X}b_j + \sum_{i\in N\cap X} p_i \\
		&\geq \sum_{1\leq j \leq |\mathcal{B}|}b_j + \sum_{i: A_i\in \mathcal{A}\backslash X}a_i + \sum_{j:B_j\in \mathcal{B}\cap X}b_j \\
		& \qquad \qquad \qquad + \left(\sum_{i: A_i\in \mathcal{A}\cap X}a_i - \sum_{j:B_j\in \mathcal{B}\cap X}b_j\right) \\
		&= \sum_{1\leq j \leq |\mathcal{B}|}b_j + \sum_{1\leq i \leq |\mathcal{A}|} = H.
		\end{align*}
		\item Assume $\overline{s}\notin X$ and $\overline{t}\in X$. In this case, the only outgoing edges with finite capacities are $(s,\overline{s})$ and $(\overline{t},t)$, and thus we get
		\begin{align*}
		\operatorname{cap}(X) \geq c(s,\overline{s}) + c(\overline{t},t) = \sum_{1\leq i \leq |\mathcal{A}|} a_i + \sum_{1\leq j \leq |\mathcal{B}|} b_j = H.
		\end{align*}
	\end{enumerate}
	Thus, a minimum $s$-$t$-cut of $G$ has at least capacity~$H$. Since there exists a cut with capacity~$H$, this bound is tight. We still need to show that the finite capacities on outgoing edges from $\overline{s}$ and ingoing edges to $\overline{t}$ are fully exhausted. Suppose not, then there exists at least one edge which is not fully exhausted. We prove this by case distinction. Suppose an edge~$(\overline{s}, i)$ from $\overline{s}$ to a player node~$i$ has flow value~$f_i$ strict less than $p_i$. Then, the capacity on outgoing edges from $s$ and $\overline{s}$ is given by $\sum_{A_i\in\mathcal{A}} a_i + \sum_{B_j\in\mathcal{B}} b_j - \sum_{i\in N}p_i$. Since we assumed at least one edge capacity is not fully exhausted, we obtain $\sum_{i\in N} p_i  < \sum_{A_i\in\mathcal{A}} a_i + \sum_{B_j\in\mathcal{B}} b_j$, which is a contradiction to $p$ being a core vector. Thus, in this case all the edges are fully exhausted. 
	
	The case for ingoing edges to $\overline{t}$ follows analogously. In total, this shows that a feasible flow~$f$ exists that induces the given payment~$p$.
\end{proof}

With the equivalence between a core vector and a feasible flow in the profit sharing graph, the problem of finding a such a core vector can now be done in polynomial time using any polynomial time maximum flow algorithm (cf.~\cite{ahuja1988network} for an overview of flow algorithms). We summarize this in the next theorem.

\begin{theorem}\label{thm: core element characterization}
	A core element for an RPS game can be computed in polynomial time. $\Box$
\end{theorem}

\section{Conclusion and Outlook}
In this paper we introduced a novel class of combinatorial cooperative games, namely the \emph{reward-penalty-selection games} (RPS games) which are based on the \emph{reward-penalty-selection problem} (RPSP). We showed that an RPS game is convex and, thus, its core is always non-empty. Furthermore, we showed that RPS games are a proper subgroup of convex games. RPS games allow a polynomial computation of the Shapley value. Focusing more on solution vectors of RPS games, we gave a characterization of core elements as feasible flows in a network graph. Thus, a core element can be computed not only in polynomial time but also efficiently.

Future research is directed to find flow representations of other payment vectors that fulfill certain fairness properties. For instance, the egalitarian allocation (cf. \cite{dutta1989concept, koster1999weighted}) has the property to distribute the obtained profit in a ``most equal'' way among the players as well as being contained in the core given a convex game. This could be transformed to finding a feasible flow of a certain value with the additional property that the flow values on edges do not differ ``too much''. This problem of finding an \emph{almost equal flow} was introduced in~\cite{haese2020algorithms}.  However, at the current stage it is not clear whether such a one-to-one correspondence between almost equal flows and egalitarian allocations actually holds.
 
Altogether, due to their generality, we think that RPS games are a widely viable modeling technique for profit sharing for several underlying application settings. 

\bibliographystyle{plain}
\bibliography{references}

\newpage
\noindent
Till Heller\\
Department of Optimization\\
Fraunhofer ITWM, Kaiserslautern\\
Germany\\
ORCiD: 0000-0002-8227-9353\\

\noindent
Niklas Gräf\\
Sven O. Krumke\\
Optimization Research Group, Department of Mathematics\\
Technische Universit\"at Kaiserslautern, Kaiserslautern\\
Germany\\
\end{document}